\documentclass[11pt,reqno]{amsart}

\pdfoutput = 1

\usepackage[english]{babel}
\usepackage{amsmath,amssymb,amsthm}
\usepackage{amsfonts}
\usepackage{amsxtra}

\usepackage{array,dcolumn}
\usepackage{graphicx}
\usepackage{float}
\usepackage[export]{adjustbox}
\usepackage{subcaption}
\usepackage[hyperfootnotes=true,
        pdffitwindow=true,
        plainpages=false,
        pdfpagelabels=true,
        pdfpagemode=UseOutlines,
        pdfpagelayout=SinglePage,
                hyperindex,]{hyperref}

\usepackage[T1]{fontenc}
\usepackage[utf8]{inputenc}
\usepackage[activate={true,nocompatibility},spacing,kerning]{microtype}
\usepackage{bm}
\usepackage{bbm}
\usepackage[sc,osf]{mathpazo}
\microtypecontext{spacing=nonfrench}

 \calclayout

\newcommand{\mathsym}[1]{{}}
\newcommand{\unicode}[1]{{}}

\theoremstyle{plain}
\newtheorem{theorem}{Theorem}
\newtheorem{lemma}[theorem]{Lemma}

\newtheorem{proposition}[theorem]{Proposition}

\theoremstyle{definition}

\theoremstyle{remark}
\newtheorem{remark}[theorem]{Remark}

\renewcommand{\leq}{\leqslant}
\renewcommand{\geq}{\geqslant}

\newcommand{\Tr}{\rm Tr}

\numberwithin{equation}{section}

\begin{document}
	\title{
		Finite size corrections at the hard edge for the Laguerre $\beta$ ensemble
	}
	\author{Peter J. Forrester and Allan K. Trinh}
	\address{Department of Mathematics and Statistics, 
ARC Centre of Excellence for Mathematical \& Statistical Frontiers,
University of Melbourne, Victoria 3010, Australia}
	\email{pjforr@unimelb.edu.au}
	\email{a.trinh4@student.unimelb.edu.au}
	\dedicatory{Dedicated to N.E.~Frankel on the occasion of his 80${}^{th}$ birthday.}
	\maketitle
	
\begin{abstract}
A fundamental question in random matrix theory is to quantify the optimal rate of convergence to universal laws.
We take up this problem for the Laguerre $\beta$ ensemble, characterised by the Dyson parameter $\beta$, and
the Laguerre weight $x^a e^{-\beta x/2}$, $x > 0$ in the hard edge limit. The latter relates to the eigenvalues
in the vicinity of the origin in the scaled variable $x \mapsto x/4N$. 
Previous work has established the corresponding functional form of various
statistical quantities --- for example the distribution of the smallest eigenvalue, provided that $a \in \mathbb Z_{\ge 0}$.
We show, using the theory of multidimensional hypergeometric functions based on Jack polynomials,
 that with the modified hard edge scaling
$x \mapsto x/4(N+a/\beta)$,  the rate of convergence to the limiting distribution is $O(1/N^2)$, which is optimal.
In the case $\beta = 2$, general $a> -1$ the explicit functional form of the distribution of the smallest eigenvalue at this order can be computed, as it can
for $a=1$ and general $\beta > 0$. An iterative scheme is presented to numerically approximate the  functional form for general
$a \in \mathbb Z_{\ge 2}$.
\end{abstract}

\section{Introduction} \label{S1}

Random matrix theory abounds in limit laws. As an example, consider an ensemble of $N \times N$ complex Hermitian matrices,
diagonal entries independently chosen from a distribution with mean $0$ and variance $1/2$, and diagonal entries
have mean zero and finite standard deviation. Assume too that all entries have finite $(2 + \epsilon)$-th moment for some
fixed $\epsilon > 0$. Let $E_N(0;(a,b))$ denote the probability that the interval $(a,b)$ is free of eigenvalues.
Then, independent of further details of the distributions,
one has the limit theorem \cite{Me91,EY17,Ag18}
\begin{equation}\label{1.1}
\lim_{N \to \infty} E_N\Big ( 0, \Big ( {\pi a \over \sqrt{2N}}, {\pi b \over \sqrt{2N}} \Big ) \Big ) = \det \Big ( \mathbb I_N - \mathbb K_{(0,b-a)}^{\rm sine} \Big ),
\end{equation}
where $ \mathbb K_{(0,b-a)}^{\rm sine}$ is the integral operator on the interval $(0,b-a)$ with kernel
$$
K(x,y) = { \sin \pi (x - y) \over \pi (x - y)}.
$$
With such a limit law established, one can ask for the rate, as a function of $N$, that the limiting distribution is approached.
Such a question is core not only from a probabilistic viewpoint on random matrix theory, but also in the context of applications of
random matrices, particularly to the Riemann zeros.

In this topic from number theory, the zeros of the Riemann zeta function $\zeta(s) = \sum_{n=1}^\infty n^{-s}$ in the upper half
complex $s$-plane are written as ${1 \over 2} + i E_p$  $(p=1,2,\dots)$ with $0 < E_1 < E_2 < \cdots $ in accordance with the
Riemann hypothesis. For large $E$ a model put forward by Keating and Snaith \cite{KS00a} predicts that the zeros in intervals of
size ${1 \over 2 \pi} \log E$ behave statistically like the eigenvalues of Haar distributed random matrices $U \in U(N)$ with
$N \approx {1 \over 2 \pi} \log E$. This is a refinement of the so-called Montgomery-Odlyzko law (see e.g.~\cite{KS99})
which asserts that for $E \to \infty$ the statistical properties of the Riemann zeros, scaled to have mean spacing unity,
coincide with the statistical properties of the bulk scaled eigenvalue of 
$N \times N$ complex Hermitian matrices with $N \to \infty$. It immediately draws attention to finite $N$ corrections to limit laws
such as (\ref{1.1}). Indeed, using the $U(N)$ model, the leading finite $N$ correction to various statistical quantities can be computed
exactly and compared against Odlyzko's high precision evaluation of large runs (over $10^9$) at very large heights ($E \approx 1.3 \times 10^{22}$)
as announced in \cite{Od01} and subsequently made available to interested researchers; see \cite{BBLM06, FM15, BFM17} for such
studies.

The analysis of finite $N$ corrections to bulk scaled statistical quantities for Hermitian matrices, in contrast to matrices from $U(N)$,
is complicated by these corrections not being translationally invariant. Thus for example one would expect that in (\ref{1.1}) a
dependence on  both end points $a$ and $b$ at this order, rather than their difference as is the case of the limit law. This means that to
obtain meaningful predictions the statistical quantity needs to be averaged over some mesoscopic interval. If one considers instead
finite $N$ corrections at the spectrum edge, this complication is not present. Appreciation of this point, together with interest and
earlier work within multivariate statistics \cite{EK06, JM12, Ma12}, motivated our recent studies \cite{FT18,FT19} at the soft
edge of various random matrix ensembles. One recalls that the defining feature of a soft edge of a random matrix spectrum is
that at the spectrum edge the spectral density is non-zero on either side. An unexpected effect was observed: in all cases analysed a scaling
and centring of eigenvalues could be found such that the leading correction to the statistical quantities being analysed occurred at order
$N^{-2/3}$. One description of this is as a weak universality, since the rate of convergence to the universal laws is independent of the details
of the random matrix model although the functional form of the leading correction term is ensemble dependent.

Random matrix ensembles can also exhibit a hard edge. Thus for example a positive definite matrix has its spectral density strictly zero for negative values,
and so if the eigenvalue support borders the origin, this spectral boundary is a hard edge. Our interest in this paper is on the form of the 
optimal scaling of eigenvalues at the hard edge for various random matrix ensembles, chosen to make the leading correction term as small as possible.
In fact there is some previous literature on this issue. Consider an $n \times N$ $(n \ge N)$ complex Gaussian matrix $X$, and form the product $X^\dagger X$.
Set $a = n - N$.
Let  $E_{N,\beta}(0;J;x^a e^{-\beta x/2})$ denote the probability that the domain $J\subset \mathbb{R}_{>0}$ contains no eigenvalues
(the reason for the argument $x^a e^{-\beta x/2}$ will become apparent later). Here $\beta$ is the so-called Dyson index, which
for complex entries equals 2. 
Edelman, 
Guionnet and P\'ech\'e \cite{EAP16} conjectured the large $N$ expansion
\begin{equation}\label{eq:12}
E_{N,2}(0;(0,s/4N);x^a e^{-x})) =  E_2^\text{hard}(s;a) + {a\over 2N} s {d\over ds}E_2^\text{hard}(s;a)
		+ O\left( {1\over N^2} \right), 
	\end{equation}
	where
	$$
E_2^\text{hard}(s;a) :=   \lim_{N \to \infty} E_{N,2}(0;(0,s/4N);x^a e^{-x})).
$$	
Proofs were subsequently provided by Perret and Schehr \cite{PS16}, and by Bornemann \cite{Bo16}.
Hachen, Hardy and Najim \cite{HHN16} gave a generalisation of (\ref{eq:12}) for the product $X^\dagger \Sigma X$,
where $ \Sigma$ is a fixed positive definite matrix satisfying certain constraints on its eigenvalues --- the form
of the $1/N$ correction persists but with $a$ renormalised.

Furthermore Bornemann \cite{Bo16} observed that the hard edge scaling can be optimally tuned to
	\begin{equation} \label{eq:13}
		\frac{s}{4N} \left(
		1 - {a\over 2N}
		\right).
	\end{equation}
	A simple Taylor series expansion says that the term proportional to $1/N$ in (\ref{eq:12}) then cancels. Therefore the leading non-trivial large $N$ correction is $O(N^{-2})$. 
	
	One immediate question is to give the functional form of this correction. We give two such expressions in Section \ref{S2}: one follows from an operator theoretic
	approach, and the other makes use of differential equations and Painlev\'e transcendents. In light of the weak universality observed at the soft edge, in the sense
	that  optimal scaling gives the same order for the leading correction term in all cases where analysis has been possible (see \cite{FT19} and references therein),
	it is natural to consider the large $N$ form of $E_{N,\beta}(0;J;x^a e^{-\beta x/2})$, and related statistical quantities,  for other values of the Dyson index $\beta$.
	Our working in Section \ref{S3}, using the theory of multivariable hypergeometric functions based on Jack polynomials to
	analyse the hard edge asymptotics of $E_{N,\beta}(0;J;x^a e^{-\beta x/2})$ for general $\beta > 0$ and $a \in \mathbb Z^+$,  shows that the extension of (\ref{eq:13}) to
\begin{equation} \label{eq:13a}
		\frac{s}{4(N+  a /\beta)}
	\end{equation}	
gives an optimal convergence to the limiting distribution at rate 	$O(N^{-2})$. Hence, as already obtained at the soft edge, there is strong evidence for
a weak universality of statistical quantities at the hard edge of random matrix ensembles: by tuning the scaling there is an optimal rate of convergence
to the limiting distribution, which is the same for all ensembles considered. This is further confirmed in Section \ref{S5}, where the methods of Section
\ref{S3} are applied to analyse the hard edge asymptotics of the spectral
	density for even $\beta$ and general $a > - 1$.
	
	In the case $\beta = 2$, and for general $a> - 1$, the results of Section \ref{S2} give the explicit functional form of the optimal correction term to the limiting hard
	edge scaled distribution of the smallest eigenvalue. A (scalar) hypergeometric representation gives the functional of the correction for general $\beta > 0$, in
	the case $a=1$. This is given in Section \ref{S3.3}. 
In Section \ref{S4}, an iterative scheme is presented to numerically approximate the functional  form of the correction for general
$a \in \mathbb Z_{\ge 2}$.

\section{Functional form of optimal correction term for complex Wishart matrices}\label{S2}
\subsection{Operator theoretic approach}
The matrix product $X^\dagger X$, where $X$ is an $(N+a)\times N$ ($a \in \mathbb Z_{\ge 0}$) matrix with complex standard Gaussian entries
is referred to as a complex Wishart matrix. The eigenvalue PDF has the functional form (see e.g.~\cite[Ch.~3]{Fo10})
	\begin{equation} \label{eq:1}
		{1\over Z_{N,2}} \prod_{l=1}^N w_2(x_l) \prod_{1\leq j<k \leq N} |x_k-x_j|^2,
	\end{equation}
	with $w_2(x)=x^a e^{-x}\chi_{x>0}$ known as the Laguerre weight. The function $\chi_A=1$ if $A$ is true and $0$ otherwise,
	while $Z_{N,2}$ denotes the normalisation. It is the exponent 2 on the product of differences which determines
	that the Dyson index here is $\beta = 2$. 
	
	For a general weight $w(x)$, a PDF of the form (\ref{eq:1}) has the special feature of its general $k$-point correlation function
	being given by a $k \times k$ determinant, fully determined by a single function of two variables $K_N(x,y)$. This kernel
	involves the orthogonal polynomials relating to the weight (see e.g.~\cite[Ch.~5]{Fo10}). In the present case
	\begin{equation} \label{eq:2}
		\rho_k(x_1,\ldots,x_k)=\det[K_N(x_i,x_j)]_{i,j=1,\ldots,k}
	\end{equation}
	where ,with $L_n^{(a)}(x)$ denoting the Laguerre polynomials,
	\begin{align} \label{eq:3}
		K_N(x,y)= (w_2(x)w_2(y))^{1/2}\sum_{n=0}^{N-1} {L_n^{(a)}(x)L_n^{(a)}(y) \over h_n}, \qquad h_n =n! \Gamma(a+n+1).
	\end{align}
	 Moreover the normalisation in (\ref{eq:1}) is given by $Z_{N,2}=\prod_{n=0}^{N-1} h_n$.
	 With use of the Christoffel-Darboux formula (see e.g.~\cite[Prop.~5.1.3]{Fo10}), and the recurrence $L_N^{(a-1)}(x)+L_{N-1}^{(a)}(x)=L_N^{(a)}(x)$, the correlation kernel (\ref{eq:3}) can be written
	\begin{equation} \label{eq:a2}
	K_N(x,y) = N! { e^{-(x+y)/2}(xy)^{a/2}\over \Gamma(a+N)} { L_N^{(a)}(x)L_N^{(a-1)}(y)-L_N^{(a-1)}(x)L_N^{(a)}(y) \over x-y }.
	\end{equation}
	
	A corollary of the determinantal structure (\ref{eq:2}) is that the gap probability $E_{N,2}(0;J;x^a e^{-x};\xi)$ (the extra argument $\xi$ on top of the
	arguments used in (\ref{eq:12}) can be thought of as a thinning parameter: each eigenvalue is deleted independently  with probability $(1 - \xi)$
	and therefore $0 < \xi \le 1$;
	see e.g.~the recent works \cite{Fo14a,FM15,BFM17,CC17,BD17,BB17} as well as the pioneering paper \cite{BP04}) has
	the operator theoretic form
	\begin{equation} \label{eq:a1}
		E_{N,2}(0;J;x^a e^{-x};\xi)) = \det( \mathbb{I} - \xi \mathbb{K}_{N,J} ),
	\end{equation}
	where $\mathbb{I}$ is the identity operator and $\mathbb{K}_{N,J}$ is the integral operator with the kernel (\ref{eq:a2}) supported on $J$. As a consequence,
	 the large $N$ expansion of (\ref{eq:a1}) for $J = (0, s/4N)$  --- referred to as a hard edge scaling --- can be deduced from knowledge of the
	 corresponding large $N$ asymptotics of (\ref{eq:a2}).
	 
	\begin{lemma}
		For general $a$ fixed, the large $N$ expansion of the scaled Laguerre polynomials is given by
		\begin{multline}  \label{eq:a3}
		(2N)^{-a}  x^{a/2} L^{(a)}_N(x/4N) = J_a(\sqrt{x}) + {1 \over N} 
		\left(
		{a \over 2} x^{1/2} J_{a-1}(\sqrt{x}) - {1\over 8} x J_{a-2}(\sqrt{x})
		\right)
		\\
		+ {1 \over N^2}\left(
		{a(a-1)\over 8} x J_{a-2}(\sqrt{x}) + \left({1\over 24}-{a\over 16} \right)x^{3/2}J_{a-3}(\sqrt{x})
		+{1\over 128}x^2 J_{a-4}(\sqrt{x})
		\right) +  O\left( {1\over N^3} \right),
		\end{multline}
		where the $J_a(x)$ are Bessel functions of the first kind. This holds uniformly in $x$ in a compact set on the positive half line.
	\end{lemma}
	\begin{proof}
		Begin with the integral representation for the Laguerre polynomials,
		\begin{equation*}
			2^a L^{(a)}_N(x/4N) = \int_\mathcal{C} {dz \over 2\pi i}
			\, e^{-xz/8N} {(z+2)^{N+a}\over z^{N+1}}
		\end{equation*}
		where $\mathcal{C}$ is a positively oriented Hankel loop contour that encloses the origin but does not contain $z=-2$. Change variables $z \mapsto 4Nz/\sqrt{x}$
		and use the elementary expansion
		\begin{equation*}
			\left(
			1 + { y\over N}
			\right)^N = e^{y} \left( 1 - {y^2\over 2N} + {3y^4 + 8y^3 \over 24N^2} + \cdots \right).
		\end{equation*}
		Then (\ref{eq:a3}) follows by recognising
		\begin{equation*}
			J_a(x) = \int_\mathcal{C} {dz \over 2\pi i} \, e^{{x\over 2}(z-1/z)} z^{-a-1}.
		\end{equation*}
		
	\end{proof}
	
	\begin{proposition}
		For large $N$, and valid uniformly for $x$ and $y$  contained in a compact set on the positive half line,
		\begin{equation}  \label{eq:a4}
		{1\over 4N}K_N \left( {x\over 4N} , {y\over 4N} \right)
		 = K_\infty^{(a)}(x,y) + {1\over N} L_1^{(a)}(x,y) + {1\over N^2}L_2^{(a)}(x,y) + O\left( {1\over N^3} \right),
		\end{equation}
		where
		\begin{align}  
		K_\infty^{(a)}(x,y)&= {\sqrt{y}J_a(\sqrt{x})J_{a}'(\sqrt{y})-\sqrt{x}J_{a}'(\sqrt{x})J_a(\sqrt{y})\over 2(x-y)}, \label{eq:a5}
		\\
		L_1^{(a)}(x,y)&= {a\over 8} J_a(\sqrt{x}) J_a(\sqrt{y}), \label{eq:a6}
		\end{align}
		and
		\begin{multline} \label{eq:a7}
		L_2^{(a)}(x,y) = -{1\over 192} \bigg[
		(a^2+x+y)J_a(\sqrt{x})J_a(\sqrt{y}) + (xy)^{1/2} J_a'(\sqrt{x})J_a'(\sqrt{y})
		\\
		-(3a^2-2)\bigg( 
		\sqrt{y}J_a(\sqrt{x})J_a'(\sqrt{y})+\sqrt{x}J_a'(\sqrt{x})J_a(\sqrt{y})
		\bigg)
		\bigg].
		\end{multline}
		In the above $J_a'(\sqrt{w}) := {d \over d u} J_a(u)  |_{u= \sqrt{w}}$.
	\end{proposition}
	\begin{proof} 
		The expansion (\ref{eq:a3}) can readily be substituted into (\ref{eq:a2}). With use of Stirling's formula for the prefactor terms, and after simplification using the Bessel recurrences
		\begin{equation*}
		{2\alpha \over u}J_\alpha(u) = J_{\alpha-1}(u)+ J_{\alpha+1}(u), \quad 2J_\alpha'(u)= J_{\alpha-1}(u)-J_{\alpha+1}(u),
		\end{equation*}
		(\ref{eq:a4}) follows.
	\end{proof}

	It can be observed that (\ref{eq:a6}) is related to (\ref{eq:a5}) by the simple derivative operation given by
	\begin{equation} \label{eq:a8}
		L_1^{(a)}(x,y) = {a\over 2}\left(
		x {\partial\over \partial x} + y {\partial\over \partial y} + 1
		\right) K_\infty^{(a)}(x,y).
	\end{equation}
	The quantity $K_\infty^{(a)}(x,y)$ is well known as the limiting hard edge kernel \cite{Fo93a}.
	The task now is to appropriately centre the hard edge scaling as in (\ref{eq:13}) so that the term in (\ref{eq:a4}) proportional to $1/N$ cancels and explicitly express the leading correction term now proportional to $1/N^2$. It should be noted that since the scaling factor is multiplicative, the term (\ref{eq:a7}) is dependent on the scaling variable up to $O(1/N^2)$. For this reason and for the ease of calculations, we work with the rescaled hard edge variable
	\begin{equation} \label{eq:a9}
		{x\over 4N+2a} = {x \over 4N} \left(
		1 - {a\over 2N} + {a^2\over 4N^2} - \cdots
		\right)
	\end{equation}
	which agrees with  (\ref{eq:13}) in the first two terms for large $N$. Making the replacement $4N \mapsto 4N + 2a$, then performing a Taylor series expansion of
	 (\ref{eq:a4}) making use of the relation (\ref{eq:a8})  eliminates the $O(1/N)$ term in (\ref{eq:a4}).

	\begin{proposition}
		For large $N$,
		\begin{equation} \label{eq:a10}
			{1\over 4N+2a} K_N\left( {x\over 4N+2a}, {y\over 4N+2a} \right) 
			=K_\infty^{(a)}(x,y) + {1\over N^2}\hat{L}_2^{(a)}(x,y) + O\left( {1\over N^3} \right),
		\end{equation}
		where
		\begin{multline} \label{eq:a11}
		\hat{L}_2^{(a)}(x,y) = -{1\over 192} \bigg(
		(a^2+x+y)J_a(\sqrt{x})J_a(\sqrt{y}) + (xy)^{1/2} J_a'(\sqrt{x})J_a'(\sqrt{y})
		\\
		+2\sqrt{y}J_a(\sqrt{x})J_a'(\sqrt{y})
		+2\sqrt{x}J_a'(\sqrt{x})J_a(\sqrt{y})
		\bigg).
		\end{multline}
	\end{proposition}

	Immediate from (\ref{eq:a10}) is that the spectral density (or the one-point correlation function) at the hard edge has the large $N$ expansion
	\begin{equation}
		{1\over 4N+2a} \rho_N \left( {x\over 4N+2a} \right)
		= \rho_{\infty,0}(x) + {1\over N^2} \hat{\rho}_{\infty,2}(x) + O\left( {1\over N^3} \right),
	\end{equation}
	with
	\begin{align}
		\rho_{\infty,0}(x) &= {1\over 4}\left( J_a(\sqrt{x})^2 - J_{a+1}(\sqrt{x})J_{a-1}(\sqrt{x}) \right)
		\\
		\hat{\rho}_{\infty,2}(x) &= -{1\over 192} \left(
		(2x+a^2)J_a(\sqrt{x})^2 + 4\sqrt{x}J_a(\sqrt{x})J_a'(\sqrt{x}) + xJ_a'(\sqrt{x})^2
		\right). \label{eq:rho2}
	\end{align}
	
	\begin{remark}
		The functional form (\ref{eq:rho2}) appears as a sum of products of a special function (i.e. the Bessel function) and its derivative. This is analogous to the known density in the soft edge for both Gaussian and Laguerre cases where the special function of interest is the Airy function (see e.g. \cite[Ch.~7]{Fo10}). 
	\end{remark}

	As shown in \cite{BFM17}, the integral operator formula (\ref{eq:12}) for a bounded interval $J = (0,s/(4N+2a))$ can be expanded in large $N$ up to its first correction. Let $\mathbb{K}_{\infty,(0,s)}$ denote the integral operator on $(0,s)$ with kernel (\ref{eq:a5}) and $\hat{\mathbb{L}}_{2,(0,s)}$ denote the integral operator on $(0,s)$ with kernel (\ref{eq:a11}). In keeping with the notation of \cite{BFM17}, define
	\begin{equation}
	\Omega( \mathbb{K}_{\infty,(0,s)} ) : \hat{\mathbb{L}}_{2,(0,s)}
	= -\det(\mathbb{I} - \mathbb{K}_{\infty,(0,s)} )
	\Tr( (\mathbb{I} - \mathbb{K}_{\infty,(0,s)} )^{-1} \hat{\mathbb{L}}_{2,(0,s)} ).
	\end{equation}
	Then we have
	\begin{equation} \label{eq:thm1}
	\det( \mathbb{I} - \xi \mathbb{K}_{N,(0,s/(4N+2a))} ) = 
	\det(\mathbb{I} - \xi \mathbb{K}_{\infty,(0,s)} )
	+ \frac{1}{N^2} \Omega( \xi \mathbb{K}_{\infty,(0,s)} ) : \xi \hat{\mathbb{L}}_{2,(0,s)}
	+ O\left( {1\over N^3} \right).
	\end{equation}
	This framework has the numerical advantage, in the methods demonstrated by Bornemann \cite{Bo10a,Bo10b}, in that the evaluation of $\det( \mathbb{I} - \xi \mathbb{K}_J )$ converges exponentially fast if the kernel is analytic in a neighbourhood of $J$. Explicit methodology for the numerical computation of $\Omega( \mathbb{K} ) : \mathbb{L}$ has been presented in \cite{BFM17} and later used in \cite{FT18}.

	\subsection{Painlev\'e theory approach}
	The operator theoretic formula (\ref{eq:a1})  can be reformulated in terms of the $\tau$-function expression
	\cite{TW94c,FW01a}
	\begin{equation} \label{eq:b1}
		\det( \mathbb{I} - \xi \mathbb{K}_{N,J} )
		= \exp\left(
		\int_0^s U_N^{(a)}(x;\xi) {dx \over x}
		\right),
	\end{equation}
	where $U_N^{(a)}$ is a solution to the $\sigma$-form Painlev\'e V differential equation
	\begin{equation} \label{eq:b2}
		(x\sigma'')^2 - \left(
		\sigma -x\sigma'+2(\sigma')^2+(a+2N)\sigma'
		\right)^2
		+4 (\sigma')^2 (\sigma'+N)(\sigma'+a+N)=0
	\end{equation}
	subject to the boundary condition
	\begin{equation} \label{eq:b3}
		U_N^{(a)}(x;\xi) \underset{x\to 0^+}{\sim} -\xi x K_N(x,x).
	\end{equation}
	
	Under the hard edge scaled variable, the expansion
	\begin{equation} \label{eq:b4}
		U_N^{(a)}\left( {x\over 4N}; \xi \right)=\sigma_0(x;\xi)+{1\over N}\sigma_1(x;\xi) + {1\over N^2}\sigma_2(x;\xi) + O\left(
		{1\over N^3}
		\right)
	\end{equation}
	is required for consistency with the results outlined in the operator theoretic approach.
	
	\begin{proposition} \label{prop4}
		Let $K_{\infty}^{(a)}$, $L_1^{(a)}$ be specified by (\ref{eq:a5}), (\ref{eq:a6}). Then $\sigma_0(x;\xi)$ satisfies the $\sigma$-form Painlev\'e III equation
		\cite{TW94b, FW01a}
		\begin{equation}\label{eq:de1}
			(x\sigma'')^2 + \sigma'(1+4\sigma')(x\sigma'-\sigma)
			= (a\sigma')^2
		\end{equation}
		subject to the boundary condition
		\begin{equation}
			\sigma_0(x;\xi) \underset{x\to 0^+}{\sim} -\xi x K_{\infty}^{(a)}(x,x)
		\end{equation}
		and $\sigma_1(x;\xi)$ satisfies the second order linear differential equation
		\begin{equation} \label{eq:de2}
			A(x)\sigma''(x) + B(x)\sigma'(x) + C(x)\sigma(x) = D(x)
		\end{equation}
		where
		\begin{align}
		A(x)&=4x^2\sigma_0''(x) \nonumber
		\\
		B(x)&=24x\sigma_0'(x)^2-16\sigma_0(x)\sigma_0'(x)+4(x-a^2)\sigma_0'(x)-2\sigma_0(x) \nonumber
		\\
		C(x)&=-2\sigma_0'(x) (4\sigma_0'(x)+1) \nonumber
		\\
		D(x)&= -a\sigma_0'(x)(x\sigma_0'(x)-\sigma_0(x))
		\end{align}
		subject to the boundary condition
		\begin{equation}
		\sigma_1(x;\xi) \underset{x\to 0^+}{\sim} -\xi x L_{1}^{(a)}(x,x).
		\end{equation}
	\end{proposition}
	\begin{proof}
	Substitute the proposed expansion given by (\ref{eq:b4}) and replace $x$ with the usual hard edge scaled variable $x/4N$ into the $\sigma$-form Painlev\'e equation (\ref{eq:b2}). Comparing the coefficients proportional to $N^2$ and $N$ gives the equations (\ref{eq:de1}) and (\ref{eq:de2}) respectively. The boundary condition follows directly from (\ref{eq:b3}).
	\end{proof}
	
	It can be verified with use of (\ref{eq:de1}) that
	\begin{equation}
		\sigma_1(x;\xi)={a\over 2} x \sigma_0'(x;\xi)
	\end{equation}
	is a solution to (\ref{eq:de2}). This is in keeping with the feature that the rescaled hard edge variable (\ref{eq:a9}) gives a correction term of $O(1/N^2)$. Therefore consider the new expansion
	\begin{equation}
	U_N^{(a)} \left( {x\over 4N+2a}; \xi \right) = \sigma_0(x;\xi) + \frac{1}{N^2}\hat{\sigma}_2(x;\xi) + O\left(
	{1\over N^3}
	\right)
	\end{equation}
	and repeat the exercise outlined below Proposition \ref{prop4} by substituting the above equation into (\ref{eq:b2}) and using the variable (\ref{eq:a9}).
	This shows that $\hat{\sigma}_2$ too satisfies a second order linear differential equation.
	
	\begin{proposition}\label{P4}
		The correction $\hat{\sigma}_2(x;\xi)$ satisfies the second order linear differential equation (\ref{eq:de2}) but with $D(x)$ replaced with
		\begin{equation}
			\hat{D}_2(x)= {1\over 8}(x\sigma_0'-\sigma_0)^2 ,
		\end{equation}
		subject to the boundary condition
		\begin{equation}
			\hat{\sigma}_2(x) \underset{x\to 0^+}{\sim} -\xi x \hat{L}_{2}^{(a)}(x,x)
		\end{equation}
		with $\hat{L}_{2}^{(a)}$ specified by (\ref{eq:a11}). Furthermore, 
		\begin{multline} \label{eq:thm2}
			E_{N,2}\left(
			0;\left(0, {s\over 4N+2a} \right); x^a e^{-x}; \xi
			\right)
			\\ =
			\exp\left(\int_0^s \sigma_0(x;\xi) \, {dx\over x}\right)
			+ {1\over N^2} \left(\int_0^s \hat{\sigma}_2(x;\xi) \, {dx\over x} \right)
			\exp\left(\int_0^s \sigma_0(x;\xi) \, {dx\over x} \right)
			+ O\left(
			{1\over N^3}
			\right).
		\end{multline}
	\end{proposition}
	
	\begin{remark}
	Results analogous to the above are known for the soft edge scaling of the Laguerre and Gaussian unitary ensembles \cite{FT18}, although as emphasised
	in the Introduction, upon optimal centring and scaling, the leading correction term is of order $N^{-2/3}$. An equation with a structure analogous to
	(\ref{eq:de2}) has also appeared in a recent study of the two-point diagonal spin-spin correlation function of the two-dimensional Ising model
	\cite{FPTW19}.
	\end{remark}
	
	Let $p_{N,\beta}(s;x^a e^{-x};\xi)$ denote the probability density function for the smallest eigenvalue in the ensemble (\ref{eq:1}) with
	$w_s(x) = x^a e^{-x} \chi_{x > 0}$. As noted below (\ref{eq:a2}) the parameter $\xi$ is a thinning parameter: each eigenvalue is deleted independently
	with probability $1 - \xi$. This quantity relates to the gap probability $E_{N,\beta}$ according to
	\begin{equation}\label{Gp}
	p_{N,\beta}(s;x^a e^{-x};\xi) = - {d \over d s}  E_{N,\beta}(0;(0,s);x^a e^{-x};\xi).
	\end{equation}
	According to the above working, this probability density for $\beta = 2$ has a well defined hard edge scaling limit, obtained by
	replacing $s$ by $s_{N,\beta}$ as specified by (\ref{eq:a9}), and the use of this variable gives an optimal rate
	of convergence at $O(1/N^2)$. Moreover the above working gives formulas implying the functional form of the correction term at this order.
	
	\begin{figure} 
		\includegraphics[width=0.7 \textwidth]{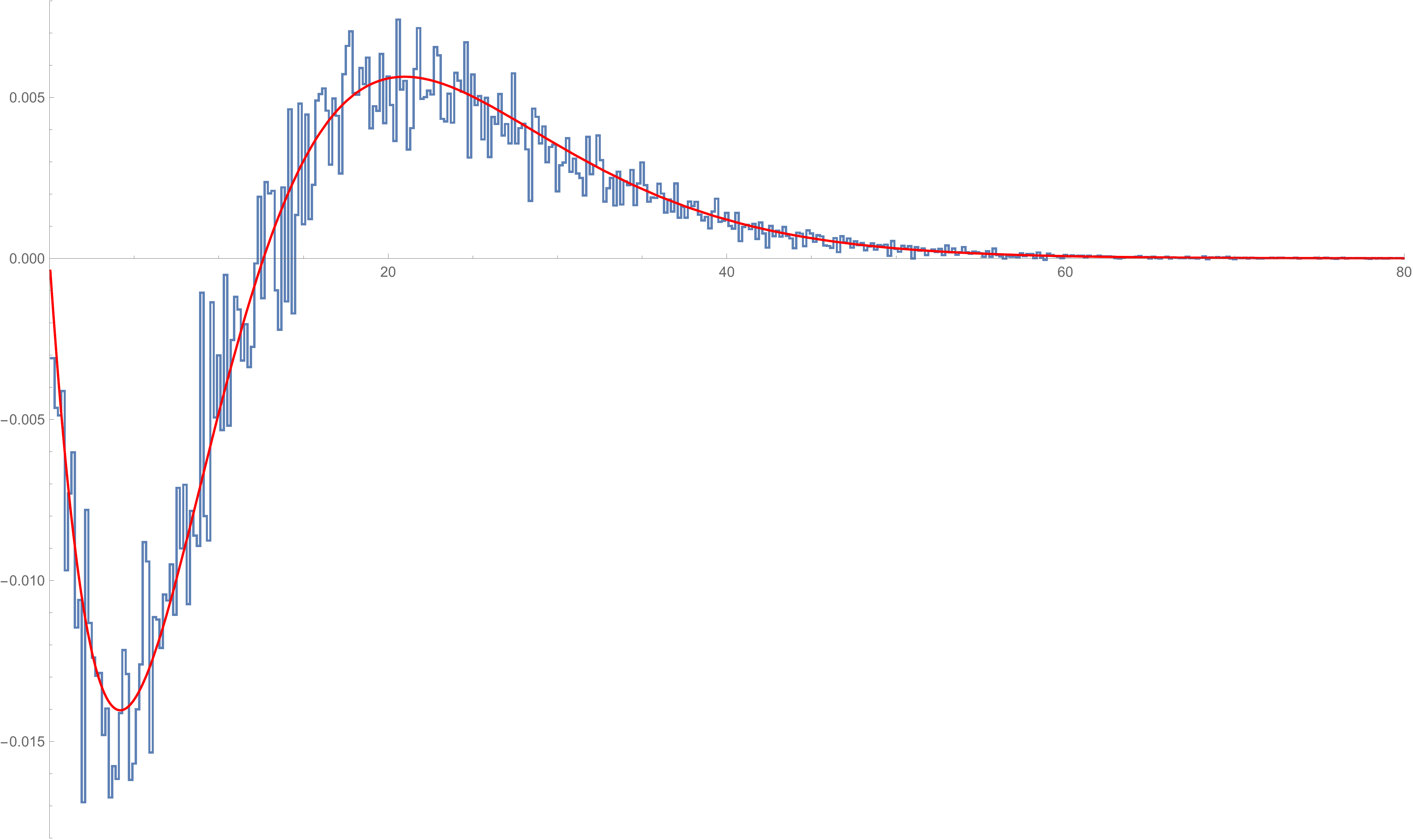}
		\caption{The leading correction term to the PDF of the smallest eigenvalue given by (\ref{eq:plot}) for $N=20$, $a=1$.  The histogram uses $500$ bins across the interval $0<s<100$ and is formed using $10^{10}$ samples.}
		\label{fig1}
	\end{figure}

	Figure \ref{fig1} is a numerical demonstration of the theoretical correction after rescaling the hard edge variables. 
	Let as abbreviate the notation $p_{N,\beta}(s;x^a e^{-x};\xi)$ in the case $\beta = 2$, $a=1$ and $\xi = 1$ as $p_N(s)$, and
	write $\lim_{N \to \infty} {\partial s_{N,2} \over \partial s} p_{N,2}(s_{N,2}) = p_{\infty}(s)$.
	The figure given by the blue steps is the difference
	\begin{equation} \label{eq:plot}
	N^2\left(  p_N^\# (s_{N,2}) - p_{\infty,0}(s) \right).
	\end{equation}
	Here, in keeping with (\ref{Gp}) and the leading term in (\ref{eq:thm1}),
	$$
	p_{\infty,0}(s) = - {d\over ds}\det(\mathbb{I} - \mathbb{K}_{\infty,(0,s)} ) \Big |_{a=1}
	$$
	is the limiting PDF of the smallest hard edge scaled eigenvalue and $p_N^{\#}(s)$ is the normalised histogram of the smallest eigenvalue from $10^{10}$ 
	numerically generated complex Wishart matrices scaled by the factor $(4N+2a)$
	with $N=20$, $a=1$.
	The red solid line is the predicted leading correction term of $O(1/N^2)$ implied by either (\ref{eq:thm1}) or (\ref{eq:thm2}) substituted in (\ref{Gp}).

\section{General $\beta$ case} \label{S3}	
\subsection{Background theory}
It was noted above that 
the matrix product $X^\dagger X$, where $X$ is an $n \times N$, $n \ge N$, matrix with complex standard Gaussian entries
has the eigenvalue PDF (\ref{eq:1})  (set $a = n -N$). In the same setting except that the entries are now real standard normals, (\ref{eq:1}) requires
modification to now read (see e.g.~\cite[Ch.~3]{Fo10})
	\begin{equation} \label{eq:1a}
		{1\over Z_{N,1}} \prod_{l=1}^N w_1(x_l) \prod_{1\leq j<k \leq N} |x_k-x_j|
	\end{equation}
	with $w_1(x)=x^{(n - N -1)/2} e^{-x/2}\chi_{x>0}$. Notice that here the Dyson index is $\beta = 1$. 
	A natural interpolation between (\ref{eq:1}) and (\ref{eq:1a}) is the PDF
\begin{equation} \label{eq:4}
	{1\over Z_{N,\beta}} \prod_{l=1}^N w_\beta(x_l) \prod_{1\leq j<k \leq N}  |x_k-x_j|^\beta
	\end{equation}
	with the weight $w_\beta(x)=x^{a} e^{-\beta x/2}\chi_{x>0}$, referred to as the Laguerre $\beta$ ensemble. For
	general $\beta > 0$ it can in fact be realised as specifying the PDF of the squared singular values of certain bi-diagonal matrices with independent
	entries \cite{DE02}.	
	
The exact evaluation of $E_{N,\beta}(0,(0,s); x^{a} e^{-\beta x/2})$ for $a \in \mathbb Z_{\ge 0}$ has been studied using methods of Selberg integral theory,
and associated special functions
 (see \cite[Ch.~13]{Fo10}). (Note that in distinction to the results of the previous section, it is not possible to include a general thinning parameter $\xi$ --- the
 results below are restricted to the case $\xi = 1$, which corresponds to no thinning.)
 To summarise the findings, introduce the generalised hypergeometric function
	\begin{equation} \label{eq:5}
		\,_{p}F_q^{(\alpha)} (a_1,\ldots,a_p;b_1,\ldots,b_q;x_1,\ldots,x_m)
		=
		\sum_{k=(k_1,\ldots,k_m)} 
		{1\over |k|!}
		{[a_1]_k^{(\alpha)} \cdots [a_p]_k^{(\alpha)} \over 
			[b_1]_k^{(\alpha)} \cdots [b_q]_k^{(\alpha)}} 
		C_k^{(\alpha)}(x_1,\ldots,x_m)
	\end{equation}
	where $k$ denotes a partition with
	\begin{equation*}
		|k| = \sum_{j=1}^m k_j, \quad
		[a]_k^{(\alpha)} = \prod_{j=1}^m {\Gamma(a-(j-1)/\alpha+k_j) \over \Gamma(a-(j-1)/\alpha)}
	\end{equation*}
	and $C_k^{(\alpha)}(x_1,\ldots,x_m)$ denotes  the renormalised symmetric Jack polynomials (see \cite[Eq.~13.1]{Fo10}). In particular when $m=1$, (\ref{eq:5}) recovers the usual definition of the classical hypergeometric function.
	
	It was shown \cite{Fo94b} (see also \cite[Proposition 13.2.6]{Fo10}) that given $a\in\mathbb{Z}_{\geq 0}$,
	\begin{equation} \label{eq:6}
		E_{N,\beta}(0,(0,s);x^{a} e^{-\beta x/2})=e^{-\beta Ns/2} \,_1 F_1^{(\beta/2)}(-N;2a/\beta;(-s)^a),
	\end{equation}
	where $\,_1F_1^{(\beta/2)}$ has the multi-dimensional integral representation
	\begin{multline} \label{eq:7}
	\,_1 F_1^{(\beta/2)}(-N;2a/\beta;(-s)^a)
	= {1\over M_a(2/\beta-1,N,2/\beta)} \int_{-1/2}^{1/2} \,dx_1 \cdots \int_{-1/2}^{1/2} \,dx_a
	\\
	\times \prod_{l=1}^a e^{2\pi i x_l (2/\beta-1)} (1+e^{-2\pi i x_l})^{N+(2/\beta-1)}
	e^{s e^{2\pi i x_l}}
	\prod_{1\leq j<k \leq a} |e^{2\pi i x_k} - e^{2\pi i x_j} |^{4/\beta}
	\end{multline}
	with
	\begin{equation} \label{eq:8}
		M_n(A,B,C) = \prod_{j=1}^n {\Gamma(1+A+B-C+jC)\Gamma(1+jC) \over \Gamma(1+A-C+jC) \Gamma(1+B-C+jC) \Gamma(1+C)}.
	\end{equation}
	Here the notation $(s)^a$ appearing in (\ref{eq:6}) denotes the $a$-tuple $(s,\ldots,s)$.

	\subsection{Hard edge scaling}
	We seek the large $N$ asymptotic expansion of (\ref{eq:7}), after replacing $s$ with the hard edge scaling $s/4N$, and under the assumption that $a$ is a positive integer. 
	Regarding the leading term, it is immediate from the definition (\ref{eq:5}), and the fact that $C_k^{(\alpha)}$ is a homogeneous function of its arguments, that
	(see \cite[Prop.~13.2.7]{Fo10})
	\begin{equation} \label{eq:9}
		E_\beta^\text{hard}(s;a)
		:=\lim_{N\to\infty} E_{N,\beta}(0,(0,s/4N);x^{a} e^{-\beta x/2})
		= e^{-\beta s/8} \,_0 F_1^{(\beta/2)}(2a/\beta;(s/4)^a).
	\end{equation}
	Moreover, beginning with (\ref{eq:7}) it is possible to express this $ \,_0 F_1^{(\beta/2)}$ function as an $a$-dimensional integral,
	\begin{multline}\label{eq:10}
		\,_0 F_1^{(\beta/2)}(2a/\beta;(s/4)^a)
		= {\Gamma(2/\beta)^a \over (2\pi)^a \Gamma(1+a)}
		\int_{-\pi}^\pi \, d\theta_1 \cdots \int_{-\pi}^\pi \, d\theta_a \,
		\prod_{l=1}^a e^{i\theta_l (2/\beta-1)} e^{s e^{i\theta_l}/4 +e^{-i\theta_l}}
		\\
		\times
		\prod_{1\leq j<k \leq a} |e^{i\theta_k} - e^{i\theta_j} |^{4/\beta}.
	\end{multline}
	This integral formula is a variant of the one given in earlier studies \cite{Fo94b}, \cite{Fo10}. In the latter the integrand involves the variable $s^{1/2}$
	instead of $s$ as above. The form given in (\ref{eq:10}) has the advantage of being better suited to the analysis of subleading terms in the large $N$
	expansion.
	
	To see how to obtain (\ref{eq:10}) from (\ref{eq:7}), and to furthermore extend the expansion beyond leading order, begin by
	changing variables $z_l = e^{2\pi i x_l}$ and then scaling $z_l \to N z_l$. Following this prescription, and using Cauchy's theorem to argue that the contour
	return the unit circle $\gamma$ (see \cite[Exercise 13.1 q4(ii)]{Fo10}) gives
	\begin{multline}
		\,_1 F_1(-N;2a/\beta; (-s/4N)^a )
		=
		{N^{a(2/\beta -1)} \over M_a(2/\beta-1,N,2/\beta)} {1\over (2\pi i)^a}
		\int_{\gamma} \, {dz_1\over z_1} \cdots \int_{\gamma} \, {dz_a\over z_a}
		\\
		\times \prod_{l=1}^a z_l^{2/\beta-1} e^{(N+2/\beta-1)\log(1+1/(N z_l))}
		e^{s z_l /4} \prod_{1\leq j<k\leq a} 
		\left( 1-{z_k\over z_j} \right)^{2/\beta} \left( 1-{z_j\over z_k} \right)^{2/\beta}.
	\end{multline}
	Expanding the logarithm then gives
	\begin{multline}\label{3.10}
		\,_1 F_1(-N;2a/\beta; (-s/4N)^a )
		=
		{N^{a(2/\beta -1)} \over M_a(2/\beta-1,N,2/\beta)} {1\over (2\pi)^a}
		\int_{-\pi}^\pi \, d\theta_1 \cdots \int_{-\pi}^\pi \, d\theta_a
		\\
		\times \prod_{l=1}^a e^{i \theta_l (2/\beta-1)} e^{s e^{i\theta_l}/4 + e^{-i\theta_l}}
		\left(
		1 + \left( {2\over \beta} -1 \right) {e^{-i\theta_l} \over N}
		- {1\over 2} {e^{-2i\theta_l} \over N } + O\left({1\over N^2}\right)
		\right) \prod_{1\leq j<k\leq a} | e^{i \theta_k} - e^{i \theta_j}|^{4/\beta}.
	\end{multline}
	Additionally the normalisation appearing in (\ref{eq:7}) can be checked via Stirling's formula to have the large $N$ form
	\begin{multline} \label{eq:c1}
		{1\over M_a(2/\beta-1,N,2/\beta)} = 
		\prod_{j=1}^a {\Gamma(2j/\beta) \Gamma(1+N+2(j-1)/\beta) \Gamma(1+2/\beta)
	\over \Gamma(N+2j/\beta) \Gamma(1+2j/\beta)	
	}
	\\
	= N^{a-2a/\beta}  {\Gamma(2/\beta)^a \over \Gamma(1+a) }
	\left(
	1 + {a^2 \over \beta N} \left( 1-{2\over \beta} \right) + O\left({1\over N^2}\right)
	\right).
	\end{multline}
	
	Substituting shows (\ref{eq:7}) exhibits the large $N$ form
	\begin{multline}\label{eq:c4}
	\,_1 F_1(-N;2a/\beta; (-s/4N)^a )
	=
	{\Gamma(2/\beta)^a \over (2\pi)^a \Gamma(1+a)}
	\int_{-\pi}^\pi \, d\theta_1 \cdots \int_{-\pi}^\pi \, d\theta_a
	\\
	\times\left(
	1 + {a^2 \over \beta N} \left( 1-{2\over \beta} \right)
	- \sum_{l=1}^a \left( 1 - {2\over \beta} \right) {e^{-i\theta_l} \over N}
	- {1\over 2} \sum_{l=1}^a {e^{-2i\theta_l} \over N } + O\left({1\over N^2}\right)
	\right)
	\\
	\times \prod_{l=1}^a e^{i \theta_l (2/\beta-1)} e^{s e^{i\theta_l}/4 + e^{-i\theta_l}}
	\prod_{1\leq j<k\leq a} | e^{i \theta_k} - e^{i \theta_j}|^{4/\beta}.
	\end{multline}
	In particular, at leading order one reads off (\ref{eq:10}).
	
	As an extension of (\ref{eq:10}), define 
	\begin{multline}\label{3.13}
		\left\langle f(\theta_1,\ldots,\theta_a) \right\rangle
		= {\Gamma(2/\beta)^a \over (2\pi)^a \Gamma(1+a)} 
		\int_{-\pi}^\pi \, d\theta_1 \cdots \int_{-\pi}^\pi \, d\theta_a
		\,
		f(\theta_1,\ldots,\theta_a)
		\prod_{l=1}^a e^{i\theta_l (2/\beta-1)} e^{s e^{i\theta_l}/4 +e^{-i\theta_l}}
		\\
		\times
		\prod_{1\leq j<k \leq a} |e^{i\theta_k} - e^{i\theta_j} |^{4/\beta},
	\end{multline}
	and denote
	\begin{equation}\label{eq:Ca1}
	C_a(s) = {a^2 \over \beta} \left( 1-{2\over \beta} \right) \langle 1\rangle
	- \left( 1 - {2\over \beta} \right) \left\langle \sum_{l=1}^a e^{-i\theta_l} \right\rangle
	- {1\over 2} \left\langle \sum_{l=1}^a e^{-2i\theta_l} \right\rangle
	\end{equation}
	as the terms proportional to $1/N$ appearing in (\ref{eq:c4}).
	
	\begin{proposition}\label{P5}
		We have
		\begin{equation} \label{eq:Ca2}
			C_a(s) = -{1\over 8} a s A_a(s) + {a\over \beta} s A_a'(s)
		\end{equation}
		with $A_a(s) = \,_0 F_1^{(\beta/2)}(2a/\beta;(s/4)^a)$.
	\end{proposition}
	\begin{proof}
		Observe that
		\begin{align}\label{eq:pf1}
			{1 \over 4} a s A_a(s) 
			&= { \Gamma(2/\beta)^a \over i (2\pi)^a \Gamma(1+a) } \int_{-\pi}^\pi \, d\theta_1 \cdots \int_{-\pi}^\pi \, d\theta_a
			\,
			\prod_{l=1}^{a} e^{i\theta_l (2/\beta-1)} e^{ e^{-i\theta_l} }
			\left(
			\sum_{j=1}^a e^{-i\theta_j} {\partial \over \partial \theta_j} \prod_{l=1}^{a} e^{s e^{i\theta_l}/4 }  
			\right) \nonumber
			\\
			& \phantom{----------------------------}
			\times \prod_{1\leq j<k \leq a} | e^{i\theta_k} - e^{i\theta_j} |^{4/\beta} \nonumber
			\\
			&= - \left( {2\over \beta} -2 \right) \left\langle \sum_{l=1}^a e^{-i\theta_l} \right\rangle
			+ \left\langle \sum_{l=1}^a e^{-2i\theta_l} \right\rangle
			-{2\over \beta} \left\langle \sum_{ \underset{j\neq k}{j,k=1} }^a
			e^{-i\theta_j } \left( { e^{i\theta_j} \over e^{i\theta_j}-e^{i\theta_k} } - { e^{-i\theta_j} \over e^{-i\theta_j}-e^{-i\theta_k} } \right)
			 \right\rangle,
		\end{align}
		where the second equality is the result of integration by parts. Note that interchanging the indices $\theta_j \leftrightarrow \theta_k$ does not affect the value of the sum in the final term in (\ref{eq:pf1}).
		With some manipulation this shows
		\begin{equation*}
			\left\langle \sum_{ \underset{j\neq k}{j,k=1} }^a
			e^{-i\theta_j } \left( { e^{i\theta_j} \over e^{i\theta_j}-e^{i\theta_k} } - { e^{-i\theta_j} \over e^{-i\theta_j}-e^{-i\theta_k} } \right)
			\right\rangle
			=
			-(a-1) \left\langle
			\sum_{j=1}^a e^{-i\theta_j}
			\right\rangle,
		\end{equation*}
		and thus
		\begin{equation}\label{eq:pf2}
			{1 \over 4} a s A_a(s) =
			- 2 \left( {2\over \beta} -1 -{a\over \beta} \right) \left\langle \sum_{l=1}^a e^{-i\theta_l} \right\rangle
			+ \left\langle \sum_{l=1}^a e^{-2i\theta_l} \right\rangle.
		\end{equation}
		One can similarly find
		\begin{align} \label{eq:pf3}
			s A_a'(s) &= { \Gamma(2/\beta)^a \over i (2\pi)^a \Gamma(1+a) } \int_{-\pi}^\pi \, d\theta_1 \cdots \int_{-\pi}^\pi \, d\theta_a
			\,
			\prod_{l=1}^{a} e^{i\theta_l (2/\beta-1)} e^{ e^{-i\theta_l} }
			\left(
			\sum_{j=1}^a {\partial \over \partial \theta_j} \prod_{l=1}^{a} e^{s e^{i\theta_l} /4}  
			\right) \nonumber
			\\
			&\phantom{----------------------------}
			\times \prod_{1\leq j<k \leq a} | e^{i\theta_k} - e^{i\theta_j} |^{4/\beta} \nonumber
			\\
			&= -a \left( {2\over \beta}-1
			\right)
			\left\langle 1 \right\rangle  + \left\langle \sum_{l=1}^a e^{-i\theta_l} \right\rangle.
		\end{align}
		The result (\ref{eq:Ca2}) then immediately follows from (\ref{eq:Ca1}), (\ref{eq:pf2}) and (\ref{eq:pf3}).
	\end{proof}

	Substituting the result of Proposition \label{P5} in (\ref{eq:c4}) shows that
	for large $N$,
	\begin{equation}
		\,_1 F_1(-N;2a/\beta;(-s/4N)^a) = A_a(s) + {1\over N} \left(
		-{a\over 8}s A_a(s) + {a\over \beta}s A_a'(s)
		\right)
		+
		O\left(
		{1\over N^2}
		\right).
	\end{equation}
	Now substituting this equation into (\ref{eq:6}) gives the extension to (\ref{eq:12}) for general $\beta$.

	\begin{theorem}
		For $a\in \mathbb{Z}_{\geq 0}$, we have
		\begin{equation} \label{eq:th1}
			E_{N,\beta}(0;(0,s/4N); x^a e^{-\beta x/2}) = 
			E_\beta^{\rm hard}(s;a) + {a\over \beta N} s {d\over ds}E_\beta^{\rm hard}(s;a)
			+ O\left( {1\over N^2} \right).
		\end{equation}
		As a consequence, with
		\begin{equation} \label{eq:sb}
		s_{N,\beta} = {s \over 4(N +  a / \beta)  }
		\end{equation}
		we have
		\begin{equation} \label{eq:th1}
			E_{N,\beta}(0,(0,s_{N,\beta}); x^a e^{-\beta x/2}) = 
			E_\beta^{\rm hard}(s;a) 
			+ O\left( {1\over N^2} \right).
		\end{equation}
		
	\end{theorem}
	
	\subsection{The case $a=1$}\label{S3.3}
	According to (\ref{eq:6}), and the fact contained in the sentence above the latter, we have for $a=1$
	\begin{equation} \label{a1}
	E_{N,\beta}(0;(0,s/4N); x e^{-\beta x/2}) = e^{-\beta s/8} \sum_{p=0}^\infty {(-N)_p \over p! (2/\beta)_p} (-s/4)^p.
	\end{equation}
	Noting that
	$$
	(-N)_p = (-1)^p N^p \Big ( 1 - {p (p-1) \over 2N} + {p \over 24 N^2} (-2 + 9p  - 10p^2 + 3 p^3 ) + 
	O \Big ( {1 \over N^3} \Big ) \Big )
	$$
	we obtain, after minor manipulation
	\begin{equation} \label{a2}
	E_{N,\beta}(0;(0,s/4N); x e^{-\beta x/2}) = V(s) + {1 \over N \beta} s {d \over d s} V(s) + {1 \over N^2} W(s) + O \Big ( {1 \over N^3} \Big ) 
	\end{equation}
	where
	$$
	V(s) = e^{-\beta s/8} {}_0 F_1\Big ( {2 \over \beta}; {s \over 4} \Big ), \quad
	W(s) = {1 \over 24} e^{-\beta s/8} \sum_{p=0}^\infty {p \over p! (2/\beta)_p} (-2 + 9p - 10p^2 + 3 p^3)
	\Big ( {s \over 4} \Big )^p.
	$$
	
	Recalling the definition (\ref{eq:sb}), it follows from (\ref{a2}) that
	\begin{equation} \label{a2a}
	E_{N,\beta}(0;(0,s_{N,\beta}); x e^{-\beta x/2}) = V(s) - {s^2 \over 2 (N \beta)^2} {d^2 \over d s^2} V(s) + {1 \over N^2} W(s) + O \Big ( {1 \over N^3} \Big ). 
	\end{equation}
	Simplifying using computer algebra gives the explicit form of the $1/N^2$ term.
	
	\begin{proposition}
	We have
	\begin{multline} \label{a3}
	E_{N,\beta}(0;(0,s_{N,\beta}); x e^{-\beta x/2}) = e^{-\beta s/8} {}_0 F_1\Big ( {2 \over \beta}; {s \over 4} \Big ) + \\
	{1 \over N^2} {s \over 48} e^{- \beta s /8} \bigg ( (-1+1/\beta) {}_0 F_1\Big ( {2 \over \beta}; {s \over 4} \Big ) +
	((1-1/\beta) +s \beta/ 8) ) {}_0 F_1\Big ( {2 \over \beta} + 1; {s \over 4} \Big ) \bigg ) + O \Big ( {1 \over N^3} \Big ) .
	\end{multline}
	\end{proposition}
	
	\begin{remark}
	In the case $\beta = 2$, $a=1$, taking minus the derivative of the term proportional to $1/N^2$ in (\ref{a3}), one obtains the
	precise functional form of the curve displayed in Figure \ref{fig1}.
	\end{remark}

	\section{Numerics beyond $\beta = 2$, general $a$, and $a=1$, general $\beta$}\label{S4}
	Recursive relations \cite{Fo93,FI10a,FR12}  from the broader theory of the Selberg integral \cite[Ch.~4]{Fo10} allow for the computation of
	averages of the form
	\begin{equation}\label{4.1}
		\left\langle
		\prod_{l=1}^N (x-x_l)^n
		\right\rangle_{L\beta E(x^a e^{-\beta x/2})},
	\end{equation}
	where $\langle \cdot \rangle_{L\beta E(x^a e^{-\beta x/2})}$ denotes the expectation with respect to the PDF (\ref{eq:4}). Specifically, 
	summarising  results presented in \cite{FT19} for the Laguerre even $\beta$ ensemble, introduce fixed parameters $(\lambda_1,\lambda,\alpha)$ and the auxillary function
	\begin{multline}
		L_p[ t^{\lambda_1} e^{-\lambda t}](x) = {p! (N-p)! \over N! } \int_{0}^\infty dt_1 \cdots \int_{0}^\infty dt_N \,
		\prod_{l=1}^N t_l^{\lambda_1} e^{-\lambda t_l} |x-t_l|^{\alpha-1}
		\\
		\times  \prod_{1\leq j<k \leq N} | t_k -t_j|^{2\lambda} e_p(x-t_1,\ldots,x-t_N),
	\end{multline}
	where $e_p(t_1,\ldots,t_N)$ are the elementary symmetric polynomials. Immediately it can be seen that
	\begin{equation}\label{4.2a}
		L_N[w_\beta(t)](x) \bigg\rvert_{\alpha=n,\lambda=\beta/2} = W_{2a/\beta,\beta,N}
		\left\langle
		\prod_{l=1}^N (x-x_l)^n
		\right\rangle_{L\beta E(x^a e^{-\beta x/2})},
	\end{equation}
	where
	\begin{equation}\label{4.2b}
	W_{2a/\beta,\beta,N} = w_{2a/\beta,\beta,N} \prod_{j=1}^{N} { \Gamma(1+2j/\beta) \Gamma(1+a+(j-1)\beta/2) \over \Gamma(1+\beta/2) },
	\quad
	w_{2a/\beta,\beta,N} = (2/\beta)^{N(a+1+\beta(N-1)/2)}
	\end{equation}
	is the normalisation appearing in (\ref{eq:4}). The notation $W_{\alpha, \beta, N}$ is taken from \cite[Eq.~(4.142)]{Fo10} where it is defined as a
	multidimensional integral. The gamma function evaluation given in (\ref{4.2b}) follows from Selberg integral
	theory; see e.g.~\cite[Prop.~4.7.3]{Fo10}.
	
	The key point for present purposes is that the integrals $L_p(x)=L_p[t^{\lambda_1} e^{-\lambda t}](x)$ satisfy the linear differential-difference equations
	\begin{equation}\label{rL}
	\lambda (N-p) L_{p+1}(x) = (\lambda(N-p)x+B_p)L_p(x) + x{d\over dx}L_p(x) - D_p x L_{p-1}(x)
	\end{equation}
	for $p=0,1,\dots, N-1$,
	where
	\begin{align*}
		B_p &= (p-N)(\lambda_1 + \alpha + \lambda(N-p-1))
		\\
		D_p &= p(\lambda(N-p)+\alpha).
	\end{align*}
	We would like to use (\ref{rL}) to obtain a graphical approximation to the $O(1/N^2)$ correction term in the generalisation of (\ref{a3}) to the
	weight $x^a e^{-\beta x/2}$ for $a > 1$, $a \in \mathbb Z^+$. 
	 From the definitions,
	\begin{multline}\label{eq:av}
		E_{N,\beta}(0;(0,s);x^a e^{-\beta x/2}) = {1\over W_{2a/\beta,\beta,N}} \int_s^\infty dx_1 \cdots \int_s^\infty dx_N \prod_{j=1}^N w_\beta (x_j) \prod_{1\leq j<k \leq N} |x_k -x_j|^{\beta}
		\\
		=
		e^{-N\beta s/2} { W_{0,\beta,N} \over W_{2a/\beta,\beta,N} }
		\left\langle
		\prod_{l=1}^N (x_l +s )^a
		\right\rangle_{L\beta E(e^{-\beta x/2})},
	\end{multline}
	where the second equality is the result of the change of variables $x_j \to x_j +s$. Therefore by setting $\lambda_1 =0$, $\lambda = \beta/2$ and using the property
	\begin{equation*}
		L_N(x) = L_0(x)\bigg\rvert_{\alpha \to \alpha +1},
	\end{equation*}
	the expectation in (\ref{eq:av}) can be computed iteratively using the recurrence (\ref{rL}) subject to the initial condition $L_0 \rvert_{\alpha=1} = W_{0,\beta,N}$.
		
	\begin{figure}
		\includegraphics[width=0.7 \textwidth]{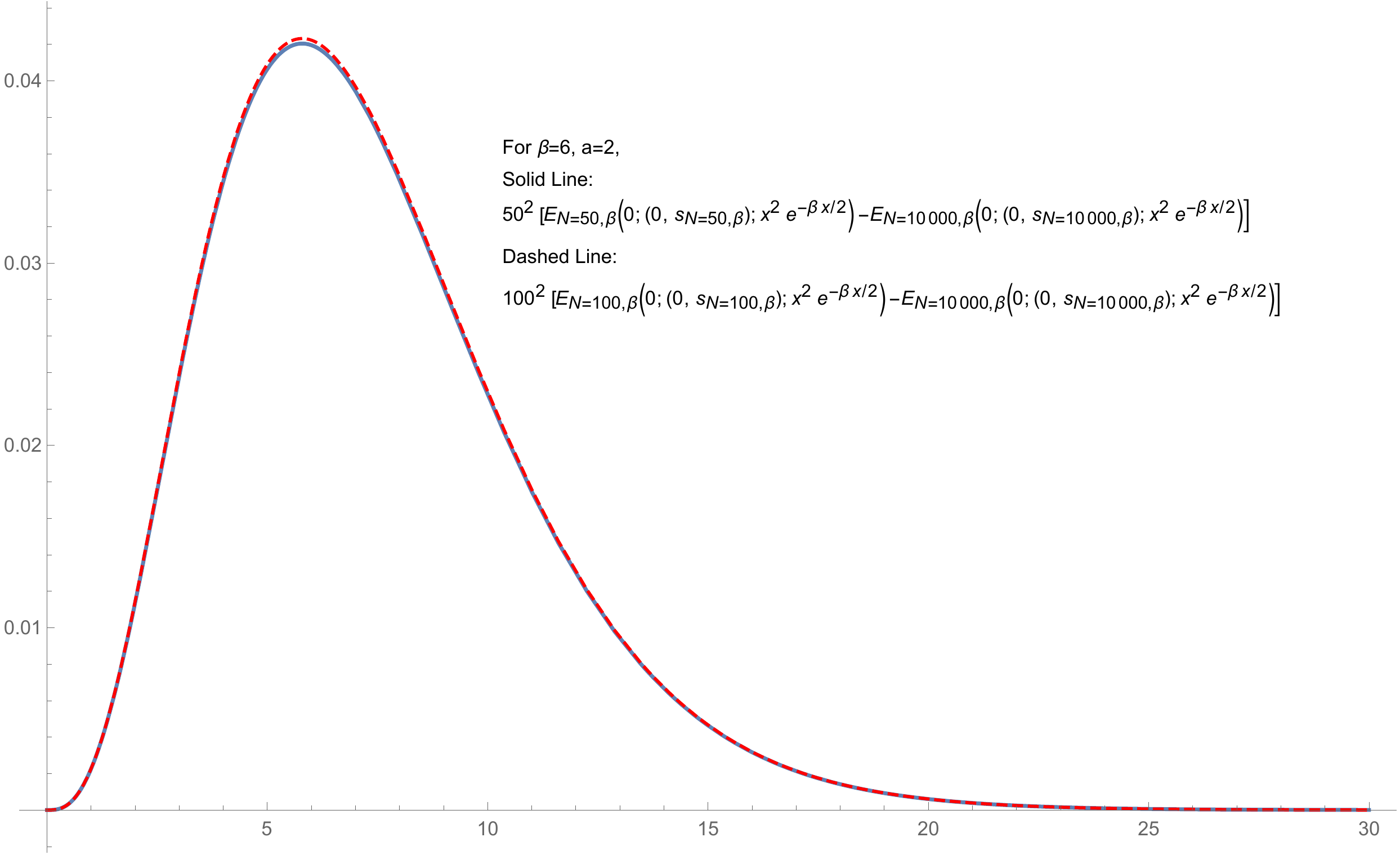}
	\caption{The leading correction term to gap probability given by (\ref{NE}) for $N=50, 100$, $a= 2$, $\beta =6$.}  \label{FigL2}
	\end{figure}

	We would like to use our numerical procedure to graph
	\begin{equation}\label{NE}
	N^2 \Big ( E_{N,\beta}(0,(0,s_{N,\beta}); x^a e^{-\beta x/2}) - E_\beta^{\rm hard}(s;a) \Big ).
	\end{equation}
	According to (\ref{eq:th1}), for large $N$ this should be of order unity and have a well defined limit. We do this by first approximating $E_\beta^{\rm hard}(s;a)$
	by $E_{N_0,\beta}(0,(0,s_{N_),\beta}); x^a e^{-\beta x/2})$, where $N_0$ is large relative to the value of $N$. Then we plot (\ref{NE}) for various
	values of $N$, expecting the final results to be almost identical. As a test, we first did this for $a=1$, $\beta = 6$. Since $a=1$, the exact value of (\ref{NE}) is
	known from (\ref{a3}). It was found that the graph obtained by following the numerical procedure with $N=50$, $N_0 = 10^4$ was almost indistinguishable
	from the theoretical curve. Having then some confidence in the accuracy of our methods, we then set out to graph (\ref{NE}) in the case
	$a = 2$, $\beta = 6$. Now there is no theoretical prediction of the limiting curve; instead different values of $N$ ($N=50$ and 100) were used
	to verify the stability of the procedure. The results are displayed in Figure \ref{FigL2}.

	\section{The spectral density for $\beta$ even}\label{S5}
	In our recent work \cite{FT19}, the fact that for even $\beta$ the spectral density of the Laguerre $\beta$ ensemble can be expressed
	as a generalised hypergeometric function based on Jack polynomials with $\beta$ variables \cite{Fo94b}, was used as the starting point for
	an analysis of the optimal correction to the {\it soft} edge scaling of this quantity. For this same quantity, and from the same
	starting point, the methods of Section \ref{S3}
	will be used here to perform an asymptotic analysis of the hard scaling, up to the order of the optimal correction.
	
	Let the spectral density in the Laguerre $\beta$ ensemble with weight $x^a e^{-\beta x/2}$ be denoted $\rho_{N,\beta} (s;a)$. We know from
\cite{Fo94b} 
\begin{equation} \label{eq:r1}
	\rho_{N,\beta} (s) = N { W_{ 2a/\beta+2, \beta, N-1 } \over W_{2a/\beta, \beta, N} }
	s^a e^{-\beta s/2} \,_1 F_1^{(\beta/2)}(-N+1; 2a/\beta +2; (s)^\beta ),  
\end{equation}
where $\,_1 F_1^{(\beta/2)}$ has the multi-dimensional integral representation
\begin{multline} \label{eq:r2}
	\,_1 F_1^{(\beta/2)}(-N+1; 2a/\beta +2; (s)^\beta ) =
	{1 \over M_\beta (2a/\beta + 2/\beta -1, N-1, 2/\beta)}
	\int_{-1/2}^{1/2} dx_1 \cdots \int_{-1/2}^{1/2} dx_\beta
	\\
	\prod_{l=1}^\beta e^{2\pi i x_l (2a/\beta + 2/\beta -1)}
	(1+e^{-2\pi i x_l})^{N + 2a/\beta + 2/\beta -2} e^{-s e^{2\pi i x_l} }
	\prod_{1 \leq j<k \leq \beta} | e^{2 \pi i x_k} -e^{2\pi i x_j}|^{4/\beta}
\end{multline}
(cf.~(\ref{eq:6})).

In relation to the normalisations, recalling (\ref{eq:8}) and (\ref{4.2b}), we see upon use of Stirling's formula and the multiplication formula for 
the gamma function that
\begin{multline} \label{eq:r3}
	{ W_{ 2a/\beta+2, \beta, N-1 } \over W_{2a/\beta, \beta, N} }
	{1 \over M_\beta (2a/\beta + 2/\beta -1, N-1, 2/\beta)}
	=
	{1\over 2\pi} \left( {4\pi \over \beta } \right)^{\beta/2}
	N^{-2+\beta-a}
	\Gamma(1+\beta/2)
	\\
	\times   \prod_{j=2}^{\beta} {\Gamma(1+2/\beta) \over \Gamma(1+2j/\beta) } \left(
	1 + {1 \over N} \left( a - {a\over \beta} - {a^2 \over \beta } \right) + O\left( {1\over N^2} \right)
	\right).
\end{multline}

Consider next the multi-dimensional integral in (\ref{eq:r2}), with the particular (but not optimal) hard edge scaling
variable corresponding to replacing $s$ by $s/4N$. Upon following the working which lead to (\ref{3.10}), this
exhibits the large $N$ expansion
\begin{multline}\label{eq:r4}
	{ N^{\beta (2a/\beta + 2/\beta -1 )} \over (2\pi)^\beta } \int_{-\pi}^\pi d\theta_1 \cdots \int_{-\pi}^\pi d\theta_\beta
	\left(
	1 + \left( {2a\over \beta} + {2\over \beta} - 2 \right) \sum_{j=1}^\beta {e^{-i\theta_j} \over N }
	- {1\over 2} \sum_{j=1}^\beta {e^{-2i\theta_j} \over N } + O\left( {1\over N^2} \right)
	\right)
	\\
	\times \prod_{l=1}^\beta e^{i \theta_l (2a/\beta + 2/\beta -1)}
	e^{-s e^{i \theta_l}/4 +e^{-i\theta_l} }
	\prod_{1 \leq j<k \leq \beta} | e^{i \theta_k} -e^{i \theta_j}|^{4/\beta}.
\end{multline}

Combining (\ref{eq:r3}) and (\ref{eq:r4}) as required by (\ref{eq:r1}) shows
\begin{equation} \label{eq:r5}
	{1\over 4N} \rho_{N,\beta} \left( {s\over 4N};a \right)
	=
	s^a  \tilde{A}_\beta(s) + {1\over N}  s^a  \tilde{C}_\beta(s)
	+ O \left( {1\over N^2} \right),
\end{equation}
where 
\begin{equation} \label{eq:r6}
	\tilde{A}_\beta (s) = {1\over 4^{1+a}} \left( {\beta \over 2} \right)^{1+2a}
	{ \Gamma(1+\beta/2) \over \Gamma(1+a) \Gamma(1+a+\beta/2) }
	\,_0 F_1^{(\beta/2)} (2a/\beta+2; (-s/4)^\beta)
\end{equation}
with
\begin{multline} \label{eq:r7}
	\,_0 F_1^{(\beta/2)} (2a/\beta+2; (-s/4)^\beta )
	=
	\left( {1\over 2\pi} \right)^\beta
	\prod_{j=1}^\beta { \Gamma(1+2/\beta) \Gamma(2a/\beta + 2j/\beta) \over \Gamma(1+2j/\beta) }
	\int_{-\pi}^{\pi} d\theta_1 \cdots \int_{-\pi}^\pi d\theta_\beta
	\\
	\times \prod_{l=1}^\beta e^{i\theta_l ( 2a/\beta + 2/\beta -1 ) } e^{-s e^{ i\theta_l}/4 + e^{-i\theta_l} }
	\prod_{1 \leq j<k \leq \beta} | e^{i \theta_k} -e^{i \theta_j} |^{4/\beta}
\end{multline}
(cf.~(\ref{eq:10})).

As a variant of (\ref{3.13}) define
\begin{multline} \label{eq:r8}
	\left\langle f(\theta_1, \ldots, \theta_\beta ) \right\rangle_\rho
	=
	c(a,\beta) \int_{-\pi}^\pi d\theta_1 \cdots \int_{-\pi}^\pi d\theta_\beta
	f(\theta_1, \ldots, \theta_\beta )
	\\
	\times \prod_{l=1}^\beta e^{i \theta_l (2a/\beta + 2/\beta -1)}
	e^{-s e^{i \theta_l}/4 +e^{-i\theta_l} }
	\prod_{1 \leq j<k \leq \beta} | e^{i \theta_k} -e^{i \theta_j} |^{4/\beta},
\end{multline}
with $c(a,\beta)$ chosen such that $\tilde{A}_\beta (s)=\left\langle 1 \right\rangle_\rho$. The $O(1/N)$ correction
term in (\ref{eq:r5}) can then be written
\begin{equation} \label{eq:r9}
	\tilde{C}_\beta (s) = 
	\left( a - {a\over \beta} - {a^2 \over \beta} - {\beta s \over 8 } \right) \left\langle 1 \right\rangle_\rho
	+ \left( {2a \over \beta} + {2 \over \beta} - 2 \right) \left\langle \sum_{l=1}^\beta e^{-i\theta_l} \right\rangle_\rho
	- {1\over 2} \left\langle \sum_{l=1}^\beta e^{-2i\theta_l} \right\rangle_\rho.
\end{equation}
In this form, we can proceed as in the derivation of Proposition \ref{P5} to relate this correction term to the leading term
in (\ref{eq:r5}). This in turn allows us to deduce that the optimal hard edge scaling results by using
the variable (\ref{eq:sb}), and as for the distribution of the smallest eigenvalue, the rate of convergence is then $O(1/N^2)$.

\begin{proposition}
	We have
	\begin{equation} \label{eq:r10}
		\tilde{C}_\beta (s) = {a \over \beta} \left( 
		s \tilde{A}_\beta '(s) + (1+a)\tilde{A}_\beta (s)
		\right)
	\end{equation}
	and thus
	\begin{equation} \label{eq:r11}
		{1\over 4N} \rho_{N,\beta} \left( {s\over 4N}; a \right)
		= 
		s^a  \tilde{A}_\beta(s) + {1\over N}  {a \over \beta} 
		{d\over ds} \left( s^{1+a}  \tilde{A}_\beta(s)  \right)
		+ O \left( {1\over N^2} \right).
	\end{equation}
	As a consequence, with $s_{N,\beta}$ given by (\ref{eq:sb}), and denoting $\rho_{\beta}^{\rm hard}(s;a) := s^a  \tilde{A}_\beta(s)$, for
	large $N$ we have
	\begin{equation} \label{eq:r11a}
	{ \partial s_{N,\beta} \over \partial s } \rho_{N,\beta} (s_{N,\beta};a)
= \rho_{\beta}^{\rm hard}(s;a) 
+ O \left( {1\over N^2} \right).
	\end{equation}
\end{proposition}

\begin{proof}
	Substituting 
	\begin{align*}
		-{1\over 4} \beta s \tilde{A}_\beta (s) 
		&= 
		\left( 4 - {2a \over \beta} - {4\over \beta} \right) \left\langle \sum_{l=1}^\beta e^{-i\theta_l} \right\rangle_\rho
		+
		\left\langle \sum_{l=1}^\beta e^{-2i\theta_l} \right\rangle_\rho ,
		\\
		s \tilde{A}_\beta '(s) 
		&= 
		-\beta \left( {2a \over \beta} + {2 \over \beta} - 1 \right) \left\langle 1 \right\rangle_\rho
		+ \left\langle \sum_{l=1}^\beta e^{-i\theta_l} \right\rangle_\rho
	\end{align*}
	into (\ref{eq:r9}) gives (\ref{eq:r10}) and (\ref{eq:r11}) follows from (\ref{eq:r10}), (\ref{eq:r5}).
	The result (\ref{eq:r11a}) follows by applying an appropriate Taylor expansion to (\ref{eq:r11}).
\end{proof}

The numerical evaluation of $\rho_{N,\beta} (s)$, in the case $\beta$ even, is straightforward for $\beta = 2$ and 4
due to there being evaluations involving only Laguerre polynomials (see e.g.~\cite{Fo10}). 
Moreover, the optimal leading correction to limiting hard edge scaled density is in the case $\beta = 2$ given by
(\ref{eq:rho2}).
Beyond these cases
the recurrence of Section \ref{S4} can be used to provide numerical evaluations, at least for moderate values of $\beta N$, as the essential task is to
evaluate (\ref{4.1}) with $N$ replaced by $N-1$ and $n = \beta$. The natural analogue of (\ref{NE}) is
\begin{equation}\label{NE1}
N^2 \bigg ( {\partial s_{N,\beta} \over \partial s}  \rho_{N,\beta} (s_{N,\beta};a) - {\partial s_{N_0,\beta} \over \partial s}  \rho_{N_0,\beta} (s_{N_0,\beta};a) \bigg ),
\end{equation}
where $N_0$ is chosen much larger than $N$. According to (\ref{eq:r11a}), for large $N$ this should have a well defined limit.
While we find that such plots are stable upon varying $N$, comparison in the case $\beta = 2$ when we have available the
exact result (\ref{eq:rho2}) indicates that these graphs do not accurately predict the limit as $s$ increases.
In this regard, note that \ref{eq:rho2}), in contrast to the analogous exact expression (\ref{a3}) for the $1/N^2$
correction in the case of the gap probability, is itself an increasing function of $s$.

	\section{Concluding remarks}
	The focus of attention in this work has been the optimal rate of convergence to the limit at the hard edge for some specific random matrix ensembles.
	The latter are the class of Laguerre $\beta$ ensembles, specified by (\ref{eq:4}) with the weight $w_\beta(x) = x^a e^{-\beta x/2} \chi_{x > 0}$.
	Using the theory of multidimensional hypergeometric functions based on Jack polynomials, it has been possible to analyse the probability that the
	interval $(0,s)$ is free of eigenvalues
	($a \in \mathbb Z^+$, general $\beta > 0$), and the eigenvalue density (even $\beta$, general $a > -1$) for large $N$. By use of the scaled variable (\ref{eq:13a})
	an optimal convergence to the limit was found at rate $O(N^{-2})$.
	
	A curious feature of (\ref{eq:13a}) is that, with $a$ replaced by $a \beta/2$ to conform to the convention used for the Laguerre
	$\beta$-ensemble in \cite{FT18,FT19} and thus with $4(N + a/\beta) \mapsto 4N + 2a$, it is this same quantity which appears in the
	optimal {\it soft} edge scaling. However it is now as an additive quantity: 
	\begin{equation}\label{e1}
	s \mapsto 4N + 2a + 2(2N)^{1/3} s. 
	\end{equation}
	
	This observation is
	suggestive from another viewpoint. As noted in the Introduction, in the case $\beta = 2$ and with $a = n - N$ the Laguerre
	$\beta$ ensemble is realised as the PDF for the squared singular values of an $n \times N$ $(n \ge N)$ standard complex Gaussian
	random matrix. Edelman, 
Guionnet and P\'ech\'e \cite{EAP16} have analysed the hard edge scaling limit of a generalisation of this model, in which the complex Gaussians are
modified by the addition of another probability with measure $d \mu(z)$. The latter is required to have mean zero, complex variance 1, and
with finite fourth moment (at least). Define the kurtosis by $\gamma = (\mu_4/ \sigma_R^4) - 3$, where $\mu_4 = \int |z|^4 \, d \mu(z)$ and
$\sigma_R^2 = \int x^2 \,  d \mu(z)$. It was proved in \cite{EAP16} that (\ref{eq:12}) holds true upon the replacement $a \mapsto a + \gamma$. It is
therefore tempting to speculate that for this same model, but now analysed at the the soft edge, an optimal convergence to the limit would again occur upon making the
replacement $a \mapsto a + \gamma$, but now in (\ref{e1}).

Similar speculations apply to the squared singular values of $\Sigma^{1/2} X$, where $X$ is an $n \times N$ $(n \ge N)$ standard complex Gaussian
	random matrix and $\Sigma$ is a fixed positive definite matrix. As remarked below (\ref{eq:12}), upon certain assumptions relating to the
	eigenvalues of $\Sigma$,
	Hachen, Hardy and Najim \cite{HHN16} proved that the effect on the $1/N$ expansion is a multiplicative renormalisation of $a=n-N$. 
	The question suggested is thus if this same multiplicative renormalisation of the parameter $a$ in (\ref{e1}) provides an
	optimal convergence rate at the soft edge.

 
\section*{Acknowledgements}
This work is part of a research program supported by the Australian Research Council (ARC) through the ARC Centre of Excellence for Mathematical and Statistical frontiers (ACEMS). PJF also acknowledges partial support from ARC grant DP170102028, and AKT acknowledges the support of a Melbourne postgraduate award.	Correspondence from Jamal Najim relating to \cite{HHN16} is much appreciated.
	
	\providecommand{\bysame}{\leavevmode\hbox to3em{\hrulefill}\thinspace}
	\providecommand{\MR}{\relax\ifhmode\unskip\space\fi MR }
	\providecommand{\MRhref}[2]{%
		\href{http://www.ams.org/mathscinet-getitem?mr=#1}{#2}
	}
	\providecommand{\href}[2]{#2}

\end{document}